\documentclass[preprint,authoryear,12pt]{elsarticle}

\usepackage{amssymb}
\usepackage{amsthm}

\usepackage{epstopdf}
\usepackage{setspace}
\usepackage{amsthm}
\usepackage{rotating}
\usepackage[T1]{fontenc}
\usepackage{subfig}
\newtheorem{definition}{Definition}[section]

\newtheorem{thm}{Theorem}[section]
\newtheorem{cor}[thm]{Corollary}

\journal{arXiv.org}
\begin{document}
\begin{frontmatter}

\title{A Comparative Study of the Signal-to-Noise Ratios of Different Representations for Symbolic Sequences}
\author{Jiasong Wang$^{1}$, Chuangyin Dang$^{2}$, Changchuan Yin$^{3,\ast}$}
\address{1.Department of Mathematics, Nanjing University, Nanjing 210093, China\\
2.Department of System Engineering and Engineering Management, City University of Hong Kong, Hong Kong\\
3.College of Natural Sciences, University of Phoenix, Chicago, IL 60173, USA\\
$\ast$ Corresponding author, Email: cyinbox@email.phoenix.edu
}
\begin{abstract}
Based on the numerical representations by T basic vectors of a symbolic sequence consisting of T symbols, first, we prove mathematical that the total Fourier spectrum of the sequence is the square of the length of the sequence. In the meantime, we define the indicator sequences vector. Using the orthogonal or row orthogonal transformations of the indicator sequences vector, we construct some special numerical representations of the symbolic sequence and characterize the signal-to-noise ratios of the power spectrum of the numerical representations. After calculating the discrete Fourier transform of those special numerical representations, the signal-to-noise ratios of them can be figured out. Mathematical theorems prove that the signal-to-noise ratio of the Fourier spectrum of those special representations of the symbolic sequence is $T/(T-1)$ times the signal-to-noise ratio of the representation by T base vectors. The results are applied in analyzing the properties of the DNA sequences or protein sequences in the frequency domain, if one uses the signal-to-noise ratios of special representations as the distinguishing criterion, the distinguishing results only depend upon the distribution of the symbols in the symbolic sequence and their mathematical constructions of representations, but do not relate to the chemical or biological meanings of the representations.
\end{abstract}
\end{frontmatter}


%
\section{Introduction}
\label{Introduction}
In the life science, a DNA sequence is a symbolic sequence of four alphabets, A, G, T, and C. A protein sequence is also a symbolic sequence constructed by 20 amino acids residues. In addition, there are other different symbolic sequences in many real world problems. To investigate the properties of those symbolic sequences by signal processing methods is a new challenge. Because a symbolic sequence does not have underlying algebraic structure, such as group structure,  symbolic signals cannot be directly processed with existing signal processing algorithms designed for signals having values that are elements of a field or a group\citep{wang2002computing}. For the purpose of studying the properties of symbolic sequence, we need to represent the symbolic sequence as a numerical sequence. Since last seventies, many researchers have suggested different numerical and graphical representations for the DNA sequences\citep{A2001}. At the fundamental level, digital signals are composed from the symbolic sequences by the use of indicator sequences. The other typical numerical representations of DNA sequences include the Voss\citep{VOSS1993}, tetrahedron\citep{SILVERMAN1986}, complex numbers \citep{A2001}, the Z-curve representations \citep{ZHANG1994,ZHANG1998}, and non-degeneracy graphical methods\citep{YAU2003, WANG2006}. Different literatures use different numerical representations, but the reasons and advantages of selecting one representation over others have not been discussed. One of important methods to analyze the symbolics sequences in signal processing is discrete Fourier transform (DFT). Bio-scientists have being studied the properties of DNA sequence or protein sequence by DFT since last nineties\citep{Tiwari1997,A2001,Dodin2000,Kotlar2003,YIN2007}.Using DFT to investigate the properties of  the symbolic sequence, a fundamental distinguishing criterion is the signal-to-noise ratio defined by DFT for the numerical representation of the symbolic sequence. We have found that the signal-to-noise ratios of some specially numerical representations for a symbolic sequence are the same, the theoretical results will be applied to the biological molecule's research.

This paper is organized as follows. It presents the numerical representations for a symbolic sequence consisting of T symbols in general, define the signal-to-noise ratio of the representation and give the theoretical proof of a formula for the signal-to-noise ratio of the numerical representation by T base vectors. Some specially numerical representations are constructed and the properties of their signal-to-noise ratios will be discussed. The interesting conclusion proved mathematically can be applied to the research of the features of a DNA sequence or a protein sequence and can also be tenable for other symbolic sequences. Because the signal-to-noise ratio for several kinds of representations of a symbolic sequence are the same, the ratio only depends upon the distribution of the symbols in the symbolic sequences and the mathematical construction of the representations.

\section{System and Methods}
\subsection{Numerical representations, DFT and the signal-to-noise ratio of a symbolic sequence }
A symbolic sequence, $s$, consists of elements $s_0 ,s_1 ,......,s_{m - 1} $, where each symbol $s_j $ belongs to a finite alphabet set $\Gamma$.  For example, $\Gamma  = \left\{ {A_1 ,A_2 , \ldots ,A_T } \right\}$, i.e., set $\Gamma$  formed by $T$  letters $A_1 ,A_2 ,......,A_T $ and $s_j $ is one of $A_1 ,A_2 ,......,A_T $.\\
A fundamental way to represent a symbolic sequence as a numerical sequence is to assign $A_t $ as a dimension $T$  base vector, $\overrightarrow{V_t }  = (0, \ldots ,1,0, \ldots ,0)^T$ whose $t$-th entrance is one, other entrances are zeros. So, they are T unit vectors of T dimension, which is one-to-one mapping with the T letters. \\
\begin{definition}
 To complete numerical representation of the symbolic sequence s we define T indicator sequences respect to the symbolic sequence $s$ as follows\\
 $$
u_{A_t } (j) = \left\{ \begin{array}{l}
 1,\;s_j  = A_t , \\
 0,{\rm{otherwise}}. \\
 \end{array} \right.\;j = 0,1, \ldots ,m - 1
$$
\end{definition}

It is obvious that a binary sequence  $u_{A_t } (j)$ with length m describes the distribution of $A_t $  in the sequence s. The symbolic sequence can be represented by the following numerical sequence.\\
$$\sum\limits_{t = 1}^T {u_{A_t } (0)\overrightarrow {V_t } } ,\sum\limits_{t = 1}^T {u_{A_t } (1)\overrightarrow {V_t } } , \ldots ,\sum\limits_{t = 1}^T {u_{A_t } (m-1) \overrightarrow {V_t } }$$
\\
\begin{definition}
 The DFT of  $u_{A_t }$   is defined as \\
$$
U_{A_t } (k) = \sum\limits_{j = 0}^{m - 1} {u_{A_t}(j)} e^{ - i2\pi kj/m} ,j = 0,1,......,m - 1
$$
\end{definition}
\begin{definition}
Using the numerical representation in Definition 2.1, the DFT of symbolic sequence $s$ is defined as\\
$$S(k) = \sum\limits_{t = 1}^T {U_{A_t }(k)} \overrightarrow {V_t } ,k = 0,1,......,m - 1$$
\end{definition}
\begin{definition}
The Fourier spectrum at frequency $k$ is defined as\\
$$P(k) = S^* (k)S(k) = \sum\limits_{t = 1}^T {U_{A_t}^2 (k)} ,k = 0,1, \ldots ,m - 1$$
for the representation of the symbolic sequence by T base vectors.
\end{definition}
\begin{definition}
The total spectrum of the symbolic sequence is defined as\\
$$P = \sum\limits_{k = 0}^{m - 1} {P(k)}$$
\end{definition}
The average Fourier spectrum E  is the total spectrum P divided by the length of the sequence, i.e. $E = P/m$. Considering frequency $k=0$ is the trivial case, then, we have the following definition:\\

\begin{definition} The signal-to-noise ratio of the symbolic sequence at the frequency $k \ne 0$ , $SNR(k)$, is defined as its the Fourier spectrum at frequency $k$ divided by its average Fourier spectrum\\
$$
SNR(k) = \sum\limits_{t = 1}^T {U_{A_t} ^2 (k)} /E,k = 1,2,.......,m - 1
$$
for the representation of the symbolic sequence by T base vectors.
\end{definition}

\begin{thm}
The total spectrum of the symbolic sequence, $P = \sum\limits_{k = 0}^{m - 1} {P(k)} $,corresponding to the representation 2.1 is ${m^{2}}$, i.e.,
$$
\sum\limits_{k = 0}^{m - 1} {\sum\limits_{t = 1}^T {U_{A_t}^2 (k} } ) = m^2
$$
\end{thm}

\begin{proof}
Suppose $S_{A_t } $  is the set of the indices for $u_{A_t}(j ) = 1,j = 0,1,......,m - 1,t = 1,2,......,T$,it is clear that $\bigcup\limits_{t = 1}^T {S_{A_t } }  = \{ 0.1,......,m - 1\} $, $\sum\limits_{t = 1}^T {\left| {S_{A_t } } \right|}  = m$ and $S_{A_i } \bigcap {S_{A_l } }  = 0,i \ne l \in \{ 1,2,......,T\} $\\
Since
$$
\begin{array}{l}
 U_{A_t} ^2 (k) = (\sum\limits_{j = 1}^{m - 1} {u_{A_t}(j)} e^{ - i2\pi jk/m} )^* (\sum\limits_{l = 1}^{m - 1} {u_{A_t}(l)e^{ - 2\pi lk/m} } ) \\
  = \sum\limits_{j\in S_{A_t } } {e^{2\pi kj/m} } \sum\limits_{l \in S_{A_t } } {e^{ - 2\pi kl/m} }  \\
  = \left| {S_{A_t } } \right| + \sum\limits_{j \ne l \in S_{A_t } } {e^{ - 2\pi (l - j)k/m} }  \\
 \end{array}
$$
Then
$$\begin{array}{l}
 \sum\limits_{k = 0}^{m - 1} {U_k ^2 (A_t )}  = m\left| {S_{A_t } } \right| + \sum\limits_{k = 0}^{m - 1} {\sum\limits_{j \ne l \in S_{A_t } } {e^{ - 2\pi (l - j)k/m} } }  \\
  = m\left| {S_{A_t } } \right| + \sum\limits_{j \ne l \in S_{A_t } } {\sum\limits_{k = 0}^{m - 1} {e^{ - 2\pi (l - j)k/m} } }  \\
 \end{array}$$
\\
Due to ${l}\neq{j}$
$$
\sum\limits_{k = 0}^{m - 1} {e^{ - 2\pi (l - j)k/m} }  = \frac{{1 - e^{ - 2\pi (l - j)m/m} }}{{1 - e^{ - 2\pi (l - j)/m} }} = 0
$$
we may see\\
$$
\sum\limits_{k = 0}^{m - 1} {U_{A_t} ^2 (k)}  = m\left| {S_{A_t } } \right|
$$
 Therefore,\\
 $$
 \begin{array}{l}
 P = \sum\limits_{k = 0}^{m - 1} {\sum\limits_{t = 1}^T {U_{A_t}^2 (k)} } ) \\
  = \sum\limits_{t = 1}^T {\sum\limits_{k = 0}^{m - 1} {U_{A_t} ^2 } } (k) \\
  = \sum\limits_{t = 1}^T {m\left| {S_{A_t}}\right|}  \\
  = m^2  \\
 \end{array}
 $$
\end{proof}
We have the following corollary.\\
\begin{cor}\label{3.1}
Using T base vectors to represent a symbolic sequence, the signal-to-noise ratio of the symbolic sequence is
$$
SNR(k) = \sum\limits_{t = 1}^T {U_{A_t} ^2 (k)} /m,k = 1,2, \ldots ,m - 1
$$
\end{cor}


\subsection{A kind of specially numerical representations and its SNR(k)}
It is clear that the indicator sequences are a kind of important representations of the distributions for the symbols in the symbolic sequence, which describes the meaningful features of the symbolic sequence. Using the indicator sequences of a symbolic sequence, we define a vector $(u_{A_1}(j),u_{A_2}(j)........u_{A_T}(j))^t,j = 0,1,......,m - 1 $ as indicator sequences vector, which is used to construct new representations of the symbolic sequence. We introduce a kind of specially numerical representations and investigate the properties of its SNR(k). The special representation is defined as follow\\
$$\left( \begin{array}{l}
 \Delta x_1 (j) \\
 \Delta x_2 (j) \\
 . \\
 . \\
 \Delta x_{T - 1} (j) \\
 1 \\
 \end{array} \right) = D\left( \begin{array}{l}
 u_{A_1}(j) \\
 u_{A_2}(j) \\
 . \\
 . \\
 . \\
 u_{A_T}(j) \\
 \end{array} \right)
$$
Where\\
$$
\begin{array}{l}
 D = \left( {\begin{array}{*{20}c}
   {a_{11} } & {a_{12} } & . & {a_{1T} }  \\
   {a_{21} } & {a_{22} } & . & {a_{2T} }  \\
   \begin{array}{l}
 . \\
 . \\
 a_{T - 11}  \\
 \end{array} & \begin{array}{l}
 . \\
 . \\
 a_{T - 12}  \\
 \end{array} & \begin{array}{l}
 . \\
 . \\
 . \\
 \end{array} & \begin{array}{l}
 . \\
 . \\
 a_{T - 1T}  \\
 \end{array}  \\
   1 & 1 & . & 1  \\
\end{array}} \right) \\
 {\rm{and}}\;j = 0,1, \ldots ,m - 1 \\
 \end{array}
$$

in which the rows of matrix $D$ are orthogonal each other and the norms of  first to the $(T-1)th$ row are the same constant d. We named the special matrix D is a row orthogonal matrix. It is evident that after normalizing the every row in the row orthogonal matrix it will obtain a diagonal matrix times a orthogonal matrix, i.e.

$$
\begin{array}{l}
 D = \left( {\begin{array}{*{20}c}
   d & 0 & . & 0  \\
   0 & d & . & 0  \\
   \begin{array}{l}
 . \\
 . \\
 0 \\
 \end{array} & \begin{array}{l}
 . \\
 . \\
 0 \\
 \end{array} & \begin{array}{l}
 . \\
 . \\
 . \\
 \end{array} & \begin{array}{l}
 . \\
 . \\
 0 \\
 \end{array}  \\
   0 & 0 & . & {\sqrt T }  \\
\end{array}} \right).D_1  \\
 {\rm{and}}\; \\
 D_1  = \left( {\begin{array}{*{20}c}
   {d_{11} } & {d_{12} } & . & {d_{1T} }  \\
   {d_{21} } & {d_{22} } & . & {d_{2T} }  \\
   \begin{array}{l}
 . \\
 . \\
 d_{T - 11}  \\
 \end{array} & \begin{array}{l}
 . \\
 . \\
 d_{T - 12}  \\
 \end{array} & \begin{array}{l}
 . \\
 . \\
 . \\
 \end{array} & \begin{array}{l}
 . \\
 . \\
 d_{T - 1T}  \\
 \end{array}  \\
   {1/\sqrt T } & {1/\sqrt T } & . & {1/\sqrt T }  \\
\end{array}} \right) \\
 \end{array}
$$
i.e., $D1$ is an orthogonal matrix. It is known that the orthogonal matrix $D1$ corresponds to a orthogonal transform. We use the first $(T-1)$ rows of the orthogonal matrix $D1$    transforms the indicator sequences vector, to construct a $(T-1)$ dimension ( $(T-1)-D$ ) numerical representation of the symbolic sequence, i.e.
$$
\begin{array}{l}
 \left( \begin{array}{l}
 \Delta x_1 (j) \\
 \Delta x_2 (j) \\
 . \\
 . \\
 \Delta x_{T - 1} (j) \\
 \end{array} \right) =  \\
 \left( {\begin{array}{*{20}c}
   {d_{11} } & {d_{12} } & . & {d_{1T} }  \\
   {d_{21} } & {d_{22} } & . & {d_{2T} }  \\
   \begin{array}{l}
 . \\
 . \\
 d_{T - 11}  \\
 \end{array} & \begin{array}{l}
 . \\
 . \\
 d_{T - 12}  \\
 \end{array} & \begin{array}{l}
 . \\
 . \\
 . \\
 \end{array} & \begin{array}{l}
 . \\
 . \\
 d_{T - 1T}  \\
 \end{array}  \\
   {} & {} & {} & {}  \\
\end{array}} \right).\left( \begin{array}{l}
 u_{A_1}(j) \\
 u_{A_2}(j) \\
 . \\
 . \\
 . \\
 u_{A_T} (j) \\
 \end{array} \right) \\
 \end{array}
$$

\begin{thm}
The signal-to-noise ratio of the representation constructed by the first $(T-1)$ rows of a orthogonal matrix and the indicator sequences vector is $T/(T-1)$ times the signal-to-noise ratio of the representation by $T$ base vectors.\\
\end{thm}
\begin{proof}
In order to figure out the signal-to-noise ratio of the $(T-1)-D$ representation, we calculate the spectrum of the representation at frequency  $k\neq0$first.\\
The DFT of $\Delta x_l (j),j = 0,1,2,.....,m - 1,$ is $\Delta X_l (k) = \sum\limits_{t = 1}^T {d_{lt} U_{A_t}(k)} ,l = 1,2,.....,T - 1$, where is $U_{A_t}(k)$ the DFT of indicator sequence, at frequency ${k\neq0}$. Therefore, the spectrum of the representation at frequency k  is
$$
\begin{array}{l}
 P(k) = \sum\limits_{l = 1}^{T - 1} {\Delta X_l ^* (k)\Delta X_l (k)}  \\
  = \sum\limits_{l = 1}^{T - 1} {(\sum\limits_{t = 1}^T {d_{lt} U_{A_t} (k)} )^* (\sum\limits_{t = 1}^T {d_{lt} U_{A_t}(k)} )}  \\
  = \sum\limits_{l = 1}^{T - 1} {[\sum\limits_{t = 1}^T {d_{lt} ^2 U_{A_t} ^2 (k) + } } 2\sum\limits_{i > j = 1}^{T - 1} {d_{lj} d_{li} {\mathop{\rm Re}\nolimits} al(} U_{A_j } (k)^* U_{A_i } (k))] \\
  = \sum\limits_{t = 1}^T {\sum\limits_{l = 1}^{T - 1} {d_{lt} ^2 U_{A_t} ^2 (k ) + } } \sum\limits_{i > j = 1}^{T - 1} 2\sum\limits_{l = 1}^{T - 1} {d_{lj} d_{li} {\mathop{\rm Re}\nolimits} al(} U_{A_j } (k)^* U_{A_i } (k)) \\
 \end{array}
$$
where $Real()$ is to take real part of a complex number. Notice $D_1$ is a orthogonal matrix, it will hold that $\sum\limits_{l = 1}^{T - 1} {d^2 _{lj} }  = (T - 1)/T,j = 1,2,......,T$ and $$\begin{array}{l}
 \sum\limits_{l = 1}^{T - 1} {d_{lj} d_{li} }  =  - 1/T\\
 {\rm{where}}\;i > j,i = 2,3,...,T,j = 1,2,...,T - 1 \\
 \end{array}$$ \\
 so $P(k)$ is transformed as\\
 $$
\begin{array}{l}
 P(k) = (T - 1)/T\sum\limits_{t = 1}^T {U_{A_t} ^2 (k)}  - 2/T\sum\limits_{i > j = 1}^{T - 1} {{\mathop{\rm Re}\nolimits} al(U_{A_j } (k)^* U_{A_i } } (k) \\
  = \sum\limits_{t = 1}^T {U^2 _{A_t } (k)}  - \frac{1}{T}(\sum\limits_{t = 1}^T {U_{A_t } (k)} )^2 \\
  = \sum\limits_{t = 1}^T {U_{A_t} ^2 (k )} ,k = 1,2, \ldots ,m - 1, \\
 \end{array}
$$
because the DFT of $\sum\limits_{t = 1}^T {u_{A_t} (j)} $ is zero.\\

The total spectrum of the $(T-1)-D$  representation is $P = \sum\limits_{l = 1}^{T - 1} {\sum\limits_{k = 0}^{m - 1} {\Delta X_l ^2 (k)} }. $  We know that the inverse DFT is %
$\Delta x_l (j) = 1/m\sum\limits_{k = 1}^{m - 1} {\Delta X_l (k)e^{i2\pi kj/m} } $
 and according to the Parseval's theorem.
 $$\begin{array}{l}
 \sum\limits_{j = 0}^{m - 1} {\Delta x_l ^2 } (j) = 1/m\sum\limits_{k = 0}^{m - 1} {\Delta X_l ^2 (k)} ,{\rm{i}}{\rm{.e}}{\rm{.,}} \\
 \sum\limits_{k = 1}^{m - 1} {\Delta X_l ^2 (k)}  = m\sum\limits_{j = 0}^{m - 1} {\Delta x_l ^2 } (j) = m\sum\limits_{j = 0}^{m - 1} {(\sum\limits_{t = 1}^T {d_{lt} u_{A_t} (j )} } )^2  \\
 \end{array}
 $$
 Therefore\\
 $$
 \begin{array}{l}
 P = \sum\limits_{l = 1}^{T - 1} {\sum\limits_{k = 0}^{m - 1} {\Delta X_l ^2 (k)} }  \\
  = m\sum\limits_{l = 1}^{T - 1} {\sum\limits_{j = 0}^{m - 1} {\sum\limits_{t = 1}^T {d_{lt} ^2 u_{A_t} ^2 (j)} } }  \\
 \end{array}
 $$
 Notice $\sum\limits_{l = 1}^{T - 1} {d^2 _{lt} }  = (T - 1)/T,t = 1,2,......,T$, the total spectrum is
 $$
 \begin{array}{l}
 P = m(T - 1)/T\sum\limits_{j = 0}^{m - 1} {\sum\limits_{t = 1}^T {u_{A_t} ^2 (j)} }  \\
  = m^2 (T - 1)/T \\
 \end{array}
 $$
 and the average spectrum $P/m = m(T - 1)/T$. From the above analysis, we obtain the signal-to-noise ratio at frequency k for the $(T-1)-D$  representation as follows
 $$
 SNR(k) = \frac{T}{{m(T - 1)}}\sum\limits_{t = 1}^T {U_{A_t} ^2 (k)} ,k = 1,2,......,m - 1
 $$
\end{proof}

Notice the difference between the row orthogonal matrix and orthogonal matrix only a diagonal matrix, whose diagonal elements are a same constant, the proof is obvious.
\begin{cor}
Using the first $(T-1)$ rows of row orthogonal matrix and indicator sequences vector to construct a numerical representation similarly $(T-1)-D$ representation for a symbolic sequence, its signal-to-noise ratio is $T/(T-1)$ times the signal-to-noise ratio of the representation by $T$ base vectors also.
\end{cor}

The conclusions of theorem 2.1 and corollary 2.4 suggest that by T base vectors to represent a symbolic sequence, one may use Fourier spectrum instead of the signal-to-noise ratio of the symbolic sequence, because the difference of the two quantities is a constant factor. The substitute may save the computational cost.

\section{Applications}
We apply the theorem 2.3 and its corollary 2.4. to analyzing the DNA sequence, a special case, the number of symbols T  is equal four. The fundament representation, the nucleotides one to one denoted by 4-D base vectors in a DNA sequence, is known as Voss representation\citep{VOSS1993}. In addition, for the analyzing the features of a DNA sequence a 3-D representation, Z-transformation representation, with chemical and /or biological meanings is following\citep{ZHANG1998}.

For any DNA sequence, we define $f_{a}$  is the number of  nucleotide $\alpha$ in the DNA sequence, length of n, $\alpha  \in \{ A,T,C,G\} $. For $f_{A},f_{G},f_{T}$ and $f_{C}$ ,set\\
$$
\begin{array}{l}
 x_1 (n) = f_A  + f_G  - f_C  - f_T , \\
 x_2 (n) = f_A  + f_C  - f_G  - f_T , \\
 x_3 (n) = f_A  + f_T  - f_C  - f_G  \\
 \end{array}
 $$
 define also that\\
 $$
\begin{array}{l}
 \Delta x_1 (n) = x_1 (n) - x_1 (n - 1), \\
 \Delta x_2 (n) = x_2 (n) - x_2 (n - 1), \\
 \Delta x_3 (n) = x_3 (n) - x_3 (n - 1) \\
 \end{array}
$$
where $\Delta x_1 ( - 1) = \Delta x_2 ( - 1) = \Delta x_3 ( - 1) = 0$, and $n=1,2,...,m-1$..  $\Delta x_1 (n),\Delta x_2 (n),\;{\rm{and}}\;\Delta x_3 (n)$ can only have the value 1 or $- 1$.  $\Delta x_1 (n)$ is equal to 1 when the $nth$ nucleotide is A or G (purine), or $- 1$ when the  $nth$ nucleotide is C or T (pyrimidine); $\Delta x_2 (n)$ is equal to 1 when the $nth$ nucleotide is A or C (amino-type),or $- 1$ when the $nth$ nucleotide is G or T (keto-type); $\Delta x_3 (n)$  is equal to 1 when the $nth$ nucleotide is A or T (weak hydrogen bond), or $- 1$when the $nth$ nucleotide is G or C (strong hydrogen bond). Therefore, a DNA sequence can be decomposed three series of digital signals, consisting of $1$ or $- 1$, each of which has clear chemical and/or biological meaning. The first series of digital signal  represents the distribution of the nucleotides of the purine/pyrimidine along the DNA sequence. The second series of digital signal represents the distribution of the bases of the amino/keto types along the DNA sequence. Similarly, the third series of digital signal   represents the distribution of the bases of the strong/weak hydrogen bonds along the DNA sequence. Notice the fact that the Z-transformation representation \citep{ZHANG1998} can be rewritten as
$$\begin{array}{l}
 \left( \begin{array}{l}
 \Delta x_1 (j) \\
 \Delta x_2 (j) \\
 \Delta x_3 (j) \\
 1 \\
 \end{array} \right) = D_z \left( \begin{array}{l}
 u_A (j) \\
 u_C (j) \\
 u_G (j) \\
 u_T (j) \\
 \end{array} \right) \\
 j = 0,1,2, \ldots ,m - 1, \\
 \end{array}$$
where   $$D_z  = \left( {\begin{array}{*{20}c}
   1 & { - 1} & 1 & { - 1}  \\
   1 & 1 & { - 1} & { - 1}  \\
   1 & { - 1} & { - 1} & 1  \\
   1 & 1 & 1 & 1  \\
\end{array}} \right)$$  is a row orthogonal matrix and its first three rows with the same norm.

It is well known that the four base vectors in Voss representation \citep{VOSS1993} formed a simplex in 4-D space and the simplex in 3-D space is a tetrahedron\citep{SILVERMAN1986}. If we choose four vertices of the tetrahedron in 3-D space to represent the four nucleotides, we may construct a new 3-D numerical representation of the DNA sequence
$$\left( \begin{array}{l}
 \Delta x(j) \\
 \Delta y(j) \\
 \Delta z(j) \\
 1 \\
 \end{array} \right) = D_5 \left( \begin{array}{l}
 u_A (j) \\
 u_T (j) \\
 u_C (j) \\
 u_G (j) \\
 \end{array} \right)
$$

where
$$
D_5 = \left( {\begin{array}{*{20}c}
   0 & {\frac{{2\sqrt 2 }}{3}} & { - \frac{{\sqrt 2 }}{3}} & { - \frac{{\sqrt 2 }}{3}}  \\
   0 & 0 & {\frac{{\sqrt 6 }}{3}} & { - \frac{{\sqrt 6 }}{3}}  \\
   1 & { - \frac{1}{3}} & { - \frac{1}{3}} & { - \frac{1}{3}}  \\
   1 & 1 & 1 & 1  \\
\end{array}} \right)
$$
is a row orthogonal matrix and the norms of its first three rows are the same also.It is clear that the two numerical representations derived from the indicator sequences vector have very different backgrounds, the Z-formation representation with chemical or biological meanings but the tetrahedron representation \citep{SILVERMAN1986} only depends on the mathematical concept. But the signal-to-noise ratios of the two representations are the same, 4/3 times the Voss representation.

To verify the theorems experimentally, we have compute the Fourier power spectra a DNA sequence, when the Voss and Z-curve maps are employed. The DNA sequence in the experiment is the exon regions of protein F56F11.4 isoform in the \emph{C.elegans} (GenBankID:$NM_171086$, length 1236 bp). The full DNA sequence with exons and introns of this gene has been used as a benchmark by many researchers to predict protein coding regions\citep{anastassiou2001genomic,yin2007prediction,jiang2008coding}. The DNA spectra shown in Figure 1 display a pronounced peak for the 3-periodicity for both Voss and Z-curve methods. The Figure 1 demonstrates that the Fourier spectra from the two different mappings preserve an equivalency by constant scale factor. The computational results are in table 1. The computational results shows the two methods generate the same DFT spectrum. The total spectra of the Voss mapping equals to square of the length, the counterpart of the Z-Curve mapping is three times square of the length. The 3-periodicity signal to the background noise ratio of the 4-D binary indicator representation is 12.7670 and the counterpart of Z-Curve is 17.0227. The ratio of the SNR of the Voss vs Z-Curve method is 4/3. The computation results are in agreement with he mathematical analysis.\\
\begin{figure}
\centering
\begin{tabular}{c}
\includegraphics[width = 3.5in]{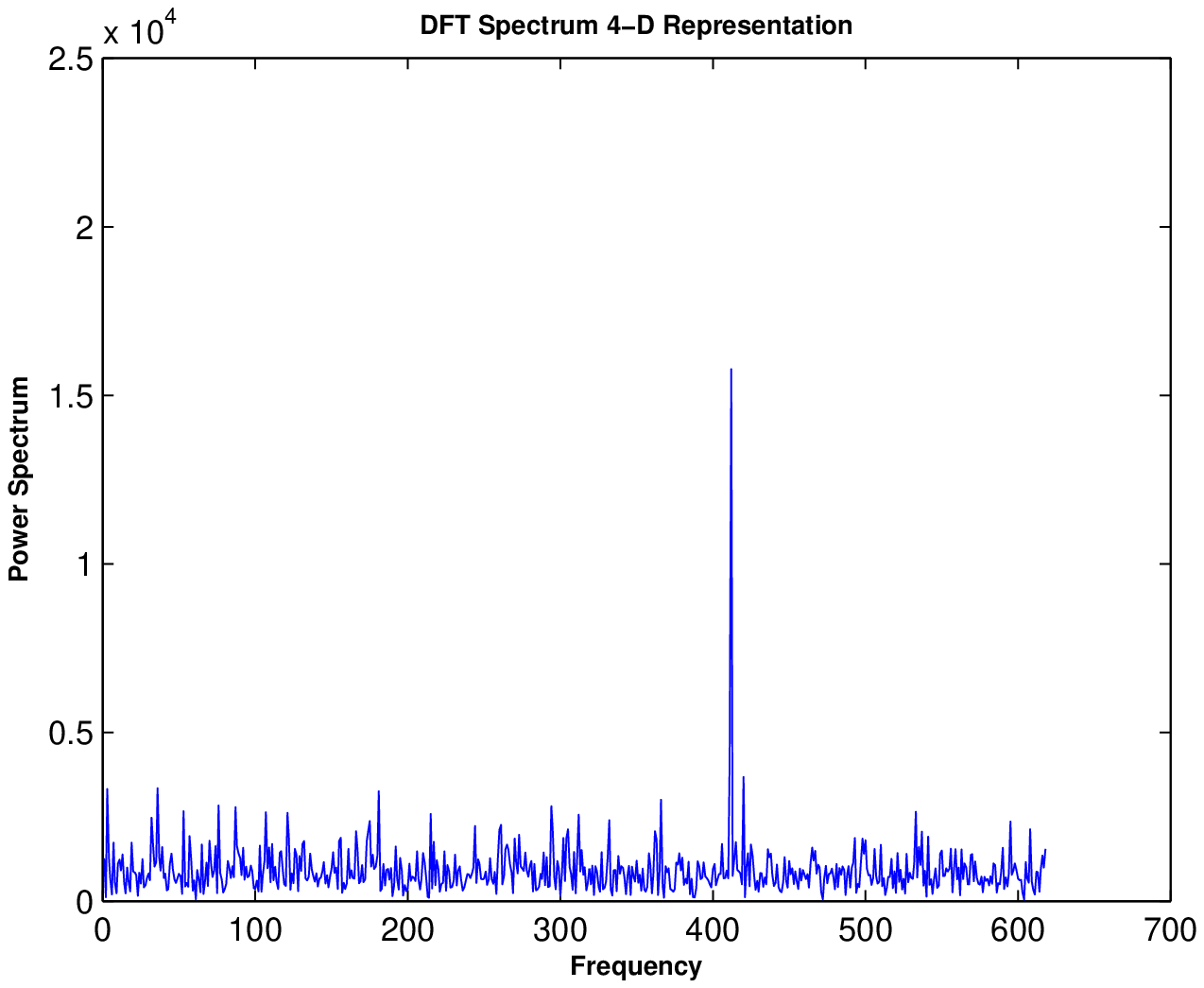} \\
\includegraphics[width = 3.5in]{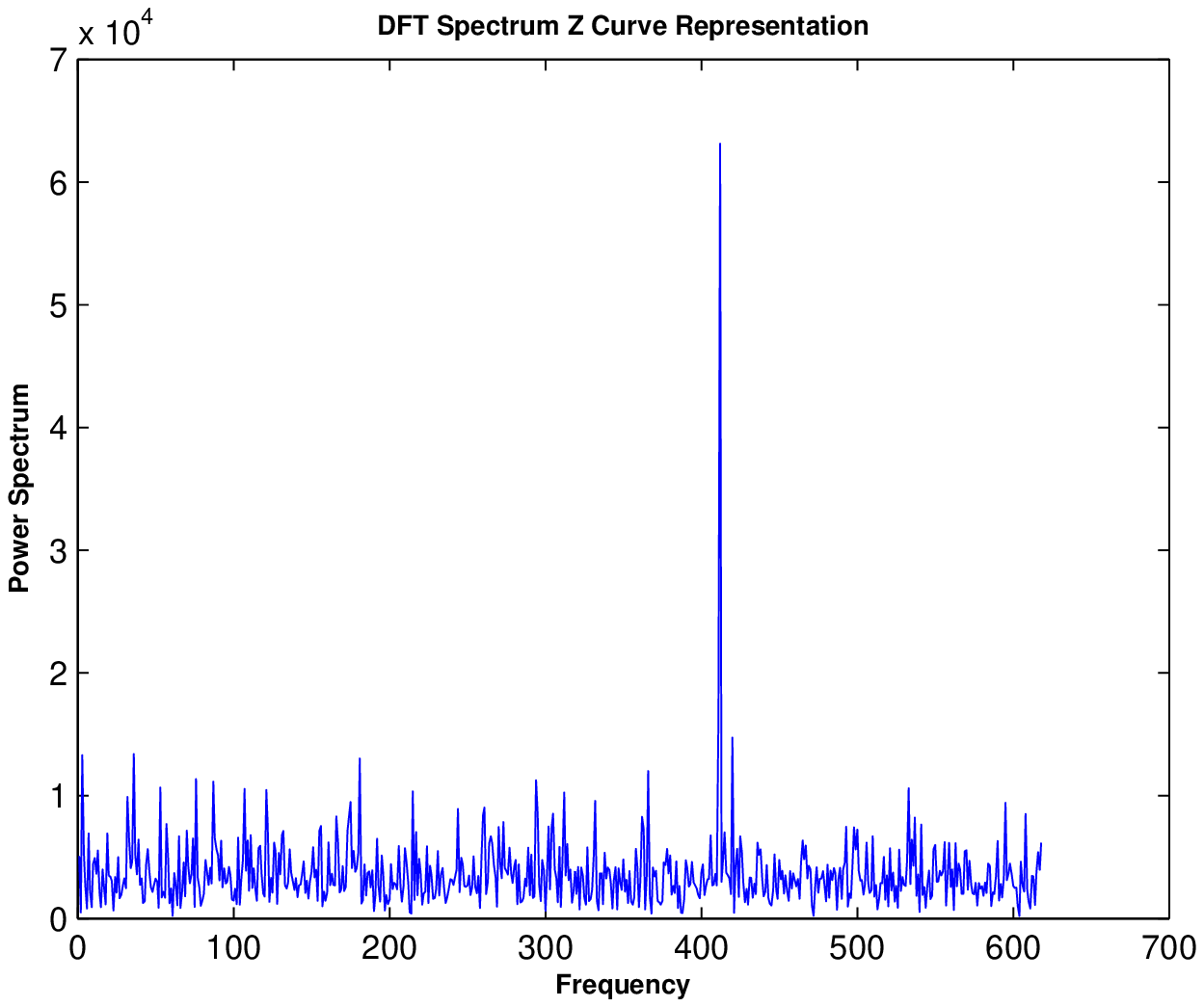} \\
\end{tabular}
\caption{DFT Spectrum of exon regions of F56F11.4 isoform in the \emph{C.elegans} (GenBankID:$NM_171086$, length 1236 bp), upper: 4-D Voss representation, low: Z-Curve representation}
\end{figure}

\begin{table}[ht]
\caption{Comparison of Fourier Spectra of the Voss and Z-Curve representations of F56F11.4 } 
\centering 
\begin{tabular}{l l l l} 
\hline\hline 
Method & Voss & Z-Curve  \\ [0.5ex] 
\hline 
    Length (bp) &  1236 & 1236 \\
    Total Spectra & 152770 & 458310 \\
    Mean Noise &  1236 & 3708 \\
    3-Periodicity & 15780 &  63120  \\
    SNR & 12.7670 &  17.0227  \\
    \hline 
\end{tabular}
\label{table:nonlin} 
\end{table}

To study the properties of a protein sequence by the method of spectrum analyzing the protein numerical encoding is fundamental step. The amino acids in protein sequence can be one-to-one mapped to 20 base vectors in 20-D space, a protein sequence then can be represented as a 20-D vector sequence. As for the Z-transform representation and tetrahedron representation, we use the row orthogonal matrix and indicator sequences vector to construct 19-D representations for the protein sequence. It is evident that the distinguishing criterions of the spectrum analyzing, their signal-to noise ratios, are $20/19$ times the one of the representations by 20 base vectors. The signal-to-noise ratios do not relate the chemical or biological properties of the numerical representations and only depend upon the mathematical construction.

\section{Conclusion}
This paper proves  mathematically first that the total spectrum of the symbolic sequence $P = \sum\limits_{k = 0}^{m - 1} {P(k)}  = m^2$  corresponding to the representation of symbolic sequence by T base vectors, i.e. the average Fourier spectrum of the symbolic sequence is the length of the sequence. Therefore, one may use Fourier spectrum instead of the signal-to-noise ratio of the symbolic sequence which may be helpful to simplify the calculation of the signal-to noise ratio.

The main contribution to the quantity of signal-to-noise ratio is the Fourier spectra of the sequences of symbols distribution in the symbolic sequence. By first $(T - 1)$ row of orthogonal matrix or row orthogonal matrix transforming the indicator sequences vector of a symbolic sequence constructs a numerical representation of the symbolic sequence, the signal-to-noise ratio of the numerical representation by T base vectors can be increased fold. It is known the increase of signal-to-noise ratio is helpful the recognition of short exon sequence in DNA sequence or the spectrum analysis of protein sequence. Noticing the applications of the theorem 2.3 and its corollary 2.4 to analyzing the DNA sequence, the signal-to-noise ratios for the Z-transformation representation and the tetrahedron representation are the same. We therefore show that the signal-to-noise of the special numerical representations constructed by orthogonal or row orthogonal transformations only depend upon the mathematical construction of representations, do not relate to the chemical or biological meanings of representations
\small

\section*{References}

\bibliographystyle{elsarticle-harv}
\bibliography{../References/myrefs}






\end{document}